\documentclass[aps, pre, 10pt, notitlepage, nofootinbib, reprint, floatfix, showpacs]{revtex4-1}

\usepackage{hyperref}
\usepackage{graphicx}
\usepackage{xcolor}
\usepackage{amssymb,amstext,amsmath,amsfonts}
\usepackage{stmaryrd}

\usepackage[caption=false]{subfig}
\usepackage{booktabs}
\usepackage[normalem]{ulem}

\usepackage{tikz}
\usetikzlibrary{shapes,backgrounds, patterns}
\usetikzlibrary{circuits.logic.US}
\usepackage{pgfplots}
\usepackage{tkz-graph}

\usepackage{amsthm}

\newtheorem{theorem}{Theorem}
\newtheorem{lemma}[theorem]{Lemma}

\newtheorem{proposition}[theorem]{Proposition}

\DeclareMathOperator*{\red}{red}
\DeclareMathOperator*{\bit}{bit}
\DeclareMathOperator*{\clsr}{cl}
\DeclareMathOperator*{\argmin}{arg\min}
\DeclareMathOperator*{\kl}{KL}

\newcommand{\Dz}{\Delta(Z)}
\newcommand{\Dx}{\Delta(X)}
\newcommand{\Imin}{I_{\min}}
\newcommand{\Ired}{I_{\red}}
\newcommand{\Ccl}{C_{\clsr}}
\newcommand{\proj}{\ssearrow}

\newcommand{\Ipr}[3]{I^{\pi}_{#3} \left({#1}\ssearrow{#2} \right)}
\newcommand{\qpr}[2]{p_{(#1  \proj #2)}}
\newcommand{\DivKL}[2]{D_{\kl}\left( #1 \,\middle\|\, #2 \right)}

\newcommand{\pia}[1]{\Pi'_{#1}}
\newcommand{\piaBW}[1]{\Pi_{#1}}
\newcommand{\iproj}[1]{{\pi_{#1}}}

\makeatletter
\@ifpackageloaded{tikz}{
  \pgfplotsset{articleplot/.style={
          axis x line=bottom,
          axis y line=left,
          axis line style={lightgray}
        }}

  \pgfplotsset{scatterplot/.style={
          articleplot,
          xlabel={$I_{\min}$},
          ylabel={$\Ired$},
          width=4.75cm,
          height=5.25cm,
          tick label style={font=\scriptsize}
        }}

  \pgfplotsset{helpplot/.style={
          lightgray
        }}

  \pgfplotsset{dataplot/.style={
          black!80,
          mark=x,
          mark size=1.5pt
        }}

  \def\uxcircle{(-1cm,-1cm) circle (1.5cm)}
  \def\syncircle{(0cm,0cm) circle (1.5cm)}
  \def\uycircle{(1cm,-1cm) circle (1.5cm)}

  \newcommand{\bidecomp}[4]{
    \bidecomphigh{#1}{#2}{#3}{#4}{10}{40}{40}{25}
  }

  \newcommand{\bidecomphigh}[8]{
      \path[use as bounding box] (-2cm,-2cm) rectangle (2cm,2cm);

      \begin{scope} 
        \clip \uxcircle;
        \fill[red!#6] \syncircle;
      \end{scope}

      \begin{scope} 
        \clip \uycircle;
        \fill[red!#7] \syncircle;
      \end{scope}
      \begin{scope} 
        \clip \uxcircle;
        \clip \uycircle;
        \fill[yellow!#5] \syncircle;
      \end{scope}
      \begin{scope}[even odd rule] 
        \clip \uxcircle (-3,-3) rectangle (3,3);
        \clip \uycircle (-3,-3) rectangle (3,3);
        \fill[cyan!#8] \syncircle;
      \end{scope}
      \begin{scope} 
        \clip \syncircle;
        \draw \uxcircle;
      \end{scope}
      \begin{scope} 
        \clip \syncircle;
        \draw \uycircle;
      \end{scope}
     \draw \syncircle;
     \draw (0cm,-1cm) node {\scriptsize #1}; 
     \draw (0cm,1cm) node {\scriptsize #4}; 
     \draw (-1cm,0cm) node {\scriptsize #2}; 
     \draw (1cm,0cm) node {\scriptsize #3}; 
  }
}
{}
\makeatother

\begin{document}

\title{A Bivariate Measure of Redundant Information}
\date{\today}
\author{Malte Harder}
\email{m.harder@herts.ac.uk}
\author{Christoph Salge}
\email{c.salge@herts.ac.uk}
\author{Daniel Polani}
\email{d.polani@herts.ac.uk}
\affiliation{Adaptive Systems Research Group\\ University of Hertfordshire}

\begin{abstract}
We define a measure of redundant information based on projections in the space of probability distributions. Redundant information between random variables is information that is shared between those variables. But in contrast to mutual information, redundant information denotes information that is shared about the outcome of a third variable. Formalizing this concept, and being able to measure it, is required for the non-negative decomposition of mutual information into redundant and synergistic information. Previous attempts to formalize redundant or synergistic information struggle to capture some desired properties. We introduce a new formalism for redundant information and prove that it satisfies all the properties necessary outlined in earlier work, as well as an additional criterion that we propose to be necessary to capture redundancy. We also demonstrate the behaviour of this new measure for several examples, compare it to previous measures and apply it to the decomposition of transfer entropy.
\end{abstract}

\pacs{02.50.-r, 89.70.Cf, 05.90.+m, 89.90.+n, 89.75.-k, 05.45.Tp}

\maketitle

\section{Introduction}
\label{sec:intro}

In this paper we present a new formalism for \textit{redundant information}; measuring for three (finite) random variables $X,Y$ and $Z$ how much information the random variable $X$ contains about $Z$ that is also contained in $Y$. \textit{Information}, in this paper, is based on Shannon entropy \cite{Shannon1948}, formalizes how much information one variable contains about another, where \textit{mutual information} is the established formalism to quantify this (see \cite{Cover2006} for a detailed account).

A naive extension of mutual information to information shared among multiple variables faces several problems. Since mutual information only measures the amount of information one variable contains about another it is unclear if two variables $X$ and $Y$, which both contain information about $Z$, actually contain the ``same'' information. Alternatively, we could ask how much additional information (e.g. reduction in entropy) about $Z$ would we get from $X$, if we already knew $Y$? This can be formalized as conditional mutual information $I(Z;X|Y)=I(Z;X,Y)-I(Z;Y)$. Thus one might think that $I(Z;X) - I(Z;X|Y)$, also called \textit{interaction information} \cite{bell2003}, is a candidate for a measure of redundant information, but the problem here is that it also captures the synergy between $X$ and $Y$ in the same measurement: in some cases, e.g. for binary variables, with $Z$ being the outcome of an \textsc{Xor} combination of $X$ and $Y$, each variable by itself contains no information about $Z$, but both taken together do contain information, which would be detected by the conditional mutual information. But we want redundant information only to be present if this information about $Z$ is present in each variable on its own. Redundant as well as synergistic information is information about the output variable contained in both variables; redundant information on the one hand is directly available in each input variable, whereas synergistic information is only available in the joint variable of the inputs. As we saw, interaction information cannot distinguish between redundant information and synergistic information, and is therefore ill-suited for this purpose.

In general, we want a redundant information formalism that quantifies how much Shannon information about the outcome of a multivariate mechanism a variable provides on its own that is also provided by all other variables as well.

\section{Related Work}

Studies of synergies and redundancies have received attention in several areas including computational neuroscience \cite{gat1999, Latham2005,Brenner2000, Balduzzi2008} and genetic regulatory networks \cite{liang2008gene,Margolin2006}. However, there seems to be no agreement how to best measure redundancy and synergy. A detailed overview of the requirements for a measure of synergy and redundancy, as well as a comprehensive overview of possible candidate measures can be found in \cite{Griffith2011}.

Generalizations of mutual information have been proposed as measures of redundant information in the literature: One of them is \textit{total correlation} also called \textit{multi-information} which measures all dependencies among the individual variables \cite{ay2006}. Another generalization is called \textit{interaction information} (as used in the introductory example in Section \ref{sec:intro}), measuring the information that is shared among the variables of the system, but not shared by any subset of the variables \cite{bell2003}. However, both measures do not explain the structure of multivariate information in terms of atomic information quantities shared between variables. The former only quantifies the dependencies, where the latter has the problem of possibly being negative. Therefore, interaction information cannot distinguish between a system of independent variables and a system where redundancies and synergies between variables compensate each other. Thus, it also fails to capture the precise structure of multivariate mutual information \cite{Williams2010, Griffith2011}.

Other measures, like \textit{interaction complexity} \cite{Kahle2009} give a good insight into the structure of interactions among random variables, however interactions and redundancy, though related, are not the same, as interaction complexity does not fulfill the criteria stated in \cite{Williams2011a}. Moreover, measures of \textit{information flow} \cite{Ay2008a,janzing2012} which are able to measure the overall amount of causal information flow, still struggle with over-determination (i.e. the measurement of redundant causal information flow), which is closely related to the problem of identifying redundant information.

A new approach addressing these problems was introduced by Williams and Beer \cite{Williams2010}. It introduces a non-negative decomposition of multivariate mutual information terms $I(Z; X_1, ... , X_k)$.
The decomposition captures all redundancies and synergies between all possible subsets of the variables $X_1,...,X_k$ with respect to another random variable $Z$. Thus, the decomposition is able to reveal the atomic structure of the information that is shared by the variables $X_1,..., X_k$ and $Z$.

Williams and Beer's decomposition can be applied to other information theoretic measures like transfer entropy as well. This allows to get further insight into the information transfer between processes by distinguishing state-independent information transfer from state dependent information transfer \cite{Williams2011}.

The information decomposition relies on a measure of redundancy \cite{Williams2010}. Redundancy quantities then become the ``building blocks'' of the construction. Information in the sense of Shannon's information theory, as used here, always denotes a measure of information that one variable contains about another. The notion of redundancy then translates to information theoretic terms as the information that two variables share about another variable.

We will argue that the redundancy measure proposed by Williams and Beer, while exhibiting a number of essential properties needed to formalize redundancy, is not capturing the concept of redundancy in a fully satisfactory way. These problems have been noted by \citet{Griffith2011}, who recently proposed \cite{Griffith2012} a synergy/redundancy measure based on {\it intrinsic conditional information} \cite{Maurer1999}, which shares similarities with an {\it information bottleneck} \cite{Tishby1999b}.

 We propose a different measure for the bivariate case which addresses our concerns and we compare it to the existing measures \cite{Williams2010,Griffith2012}. The measure is based on a geometric argument and we will show that it fulfils all axioms required by Williams for a redundancy measure \cite{Williams2011a}. We also demonstrate that the non-negativity of the information decomposition is still guaranteed when using our measure. Furthermore, we will argue in favour of an additional axiom that any measure of redundancy has to fulfil.

\subsection{Minimal Information as a Measure of Redundancy}

As mentioned above, the term redundancy has been used in several contexts denoting different quantities. Here, we specificly consider information about another random variable that is shared among several random variables and we mean the same ``piece'' of information. A candidate measure for this quantity is called \textit{minimal information} and denoted by $\Imin$ \cite{Williams2010}.

Given a set of finite random variables $X_{V}=\{X_1,...,X_n\}$, the index set $V=\{1,...,n\}$ and a finite random variable $Z$ with values from $\mathcal{X}_1 \times ... \times \mathcal{X}_n$ and $\mathcal{Z}$ respectively, we denote the mutual information between $Z$ and $X_{V} $ as follows:
\begin{equation} I(Z;X_{V}) := I(Z;X_1, ..., X_n). \end{equation}
Following \cite{Williams2010}, we now define the (non-negative) \textit{specific information} \cite{deweese1999measure}, the increase in likelihood (or reduction in surprise) of the outcome of a specific event, where $A \subseteq V$, by
\begin{eqnarray} I(Z=z;A) &:=& \sum_{x_A} p(x_A|z) \left[\log \frac{1}{p(z)} - \log  \frac{1}{p(z|x_A)}\right] \\ &=& \DivKL{p(x_A|z)}{p(x_A)}, \end{eqnarray}
where $\DivKL{\cdot}{\cdot}$ is the usual Kullback-Leibler divergence. This is then be used by Williams and Beer to define the \textit{minimal information} a set of random variables contains about the outcome as
\begin{equation} \Imin(Z;A_1,...,A_k) := \sum_{z} p(z) \min_{i} I(Z=z;A_i). \end{equation}
This measure is obviously non-negative and, in fact, positive if all variables contain some information about a specific outcome (for outcomes having probabilities which do not vanish).

For the bivariate case we will change the notation slightly and use the random variables directly instead of the index set notation, so instead of $\Imin(Z;A_1,A_2)$, where $A_1$ and $A_2$ are index sets of some collection of random variables, we will directly write $\Imin(Z;X,Y)$.

\subsection{Redundancy Axioms}
\label{sec:redundancy_axiom}

In \cite{Williams2011a}, Williams states three axioms any redundancy measure has to fulfill. For any redundancy measure $I_\cap(Z;A_1,...,A_k)$ the following must hold:
\begin{description}
  \item[Symmetry] $I_\cap$ is symmetric with respect to the $A_i$'s.
  \item[Self Redundancy] $I_\cap(Z;A)$ = $I(Z;X_A)$.
  \item[Monotonicity] \begin{equation*}
    I_\cap(Z;A_1,...,A_{k-1},A_{k}) \leq I_\cap(Z;A_1,...,A_{k-1})
  \end{equation*}
  with equality if $A_{k-1} \subseteq A_{k}$.
\end{description}
From these axioms follows the non-negativity of the redundancy measure, and that it is bounded above by the mutual information between Z and each source. To prove this, note that $A_i$ are subsets of $V$ that could be empty, and for consistency $I_\cap(Z;\emptyset)=0$ by definition. It is easy to check that all three axioms are fulfilled by the measure $\Imin$ \cite{Williams2011a}.

\subsection{Why Minimal Information is not Capturing Redundancy}

This measure contradicts a basic intuition about redundancy. Let us consider the case with two binary input variables $X,Y$ (i.e. $\mathcal{X}=\mathcal{Y}=\{0,1\}$) that are independent, uniformly distributed and where $Z=(X,Y)$ is an unaltered copy of both variables, i.e. the joint distribution of $X$ and $Y$. Now we expect that there should be no redundancy between $X$ and $Y$ with regard to $Z$ because we know that $X$ and $Y$ are independent, so the information contained about $Z$ in $X$ and $Y$ respectively is clearly not the same. However, we have $\Imin(Z;X,Y)=1 \bit$.

This happens because for each outcome of $X$ or $Y$ we observe a reduction of entropy regarding an outcome $z$ (i.e. the specific information between $X$ and $z$ as well $Y$ and $z$ is positive). However, we ignore that even though $X$ and $Y$ give the same amount of information about an outcome $z$, they tell something different about the change of the distribution $p(z)$ after an observation in $X$ or $Y$ has been made. In this particular example $X$ gives information about the first component of $Z$ while, $Y$ gives information about the second component of $Z$.

More precisely the \textit{a posteriori} distributions of $Z$, $p(z|x)$ and $p(z|y)$, when either $X$ or $Y$ have been observed, give a  different kind of information (have different content) even though they give the same amount of information. The core idea therefore is to separate the contributions of $X$ and $Y$ by adopting a geometric view in the space of probability distributions over $Z$.

\section{A New Measure of Redundant Information}

To define a new (bivariate) redundancy measure we will take a geometric view on informational quantities. Information geometry is a powerful tool-set to investigate information theoretic question in the context of Riemannian manifolds \cite{Amari2001,amari2007methods}. Geometric arguments and algorithms have profound application to information theory, statistics \cite{Csiszar2004} and have been successfully employed to construct information theoretic multivariate interaction measures \cite{Kahle2009}. Information geometry deals with statistical manifolds of probability distributions equipped with the Fisher metric \cite{amari2007methods}. The Kullback-Leibler divergence is now a divergence function on the statistical manifold and thus certain helpful properties and theorems, such as the Pythagorean Theorem, can be used. Here, we will introduce concepts of information geometry only as needed as most arguments can be done on an ad-hoc basis.

\subsection{Additional Axiom}

Before we start with the construction of the measure, we want to address the shortcoming identified above. For this purpose, we propose to add an additional axiom to the axioms from Section~\ref{sec:redundancy_axiom}. We call it the \textit{identity property}, as it states how redundancy should behave with respect to a joint random variable of identical copies of the two source variables.  It requires that for any redundancy measure $I_\cap$
\begin{equation} I_\cap\left((X_{A_1},X_{A_2});A_1,A_2\right) = I(X_{A_1};X_{A_2})\end{equation}
The idea behind this additional axiom is, that if the (bivariate) mechanism we are considering is just copying the input, the redundancy must be exactly the mutual information between the variables. Given a multivariate redundancy measure the monotonicity automatically states that the multivariate redundancy is then bounded above by the minimum of pairwise mutual information terms.

\subsection{Construction of a Redundant Information Measure}

The redundancy measure we will construct is based on the notion of \textit{projected information} which we will introduce shortly. We will begin with the definition of a bivariate redundancy measure $\Ired$, i.e. we will measure the redundancy between two sources $X$ and $Y$ with respect to $Z$ denoted by $\Ired(Z;X,Y)$.

\subsubsection{Preliminaries}

In what follows, let $\Dz$ denote the space of all probability distributions over $Z$. An information projection is now defined as the minimization of the Kullback-Leibler divergence between a probability distribution in $p \in \Dz$ and a subset $B\subset\Dz$:
\begin{equation} \iproj{B}(p) :=\argmin_{r\in B} \DivKL{p}{r}. \label{eq:proj} \end{equation}
The Kullback-Leibler divergence is not symmetric, therefore it is possible to define a dual projection $\iproj{B}^*(p)$ where the parameters of $\DivKL{\cdot}{\cdot}$ are reversed (in \cite{csiszar2003}, $\iproj{B}(p)$ is called reverse information projection and $\iproj{B}^*(p)$ information projection). Here we will exclusively use the projection $\iproj{B}(p)$.

For $B\subset \Dz$, we denote the convex closure of $B$ in $\Dz$ by \begin{equation} \Ccl(B) = \{ \lambda p + (1-\lambda)q |\, p,q \in B, \lambda \in [0,1]\}.\end{equation}
As $\Dz$ is convex we have $\Ccl(B) \subseteq \Dz$. Observing an event $x$ in $X$ or $y$ in $Y$ leads to a distribution over $Z$, $p(\cdot|x)\in \Dz$ and $p(\cdot|y) \in \Dz$ respectively. Let 
\begin{equation}
  \langle X \rangle_Z := \{p(\cdot|x) : x \in \mathcal{X}\}
\end{equation}
denote the set of all conditional distributions of $Z$ for the different events of $X$. Because the marginal distributions over $Z$ are a convex combination of the conditional distributions, namely
\begin{equation}p(z) = \sum_{x} p(z|x) p(x),\end{equation}
we have that the space of distributions over $X$, i.e. $\Dx$, is embedded in $\Dz$ by the convex set
\begin{equation} \Ccl(\langle X \rangle_Z) = \Ccl \left(\{p(\cdot|x) : x \in \mathcal{X}\}\right).\end{equation}
The convex closure of $\langle X \rangle_Z$ in $\Dz$ now contains all possible marginals $p(z)$ if we do not know the actual distribution of $X$, but where the mechanism (the conditional distribution) is known. For example, the problem of finding the channel capacity between two random variables $X$ and $Z$ can now be translated to find the point in the convex closure that maximizes its Kullback-Leibler divergence from all extremal points $p(\cdot|x)$ of the convex closure (weighted by the respective probabilities $p(x)$), as this is equivalent to maximizing the mutual information between $X$ and $Z$.


\subsubsection{Projective Information}

Using information projections we can now project the conditionals of one variable onto the convex closure of the other. We denote this  projection by
\begin{equation} \qpr{x}{Y}(\cdot) := \iproj{\Ccl(\langle Y\rangle_Z)}(p(\cdot|x)). \end{equation}
The projection is not guaranteed to be unique (for uniqueness, the set we are projecting onto would need to be log-convex and not convex \cite{csiszar2003}), however this does not matter for our purposes as we will see in the next lemma.
Now, we define the \textit{projected information} of $X$ onto $Y$ with respect to $Z$ as
\begin{equation} \Ipr{X}{Y}{Z} := \sum_{z,x} p(z,x) \log \frac{\qpr{x}{ Y}(z)}{p(z)}. \label{eq:projective_info}\end{equation}
The rationale behind this construction is that the projected information quantifies the amount of information that two variables share with each other, here $X$ and $Z$, that can be expressed in terms of the information $Y$ shared with $Z$ (we are projecting onto $Y$). This is illustrated for binary input variable in FIG.~\ref{fig:projective_info}.

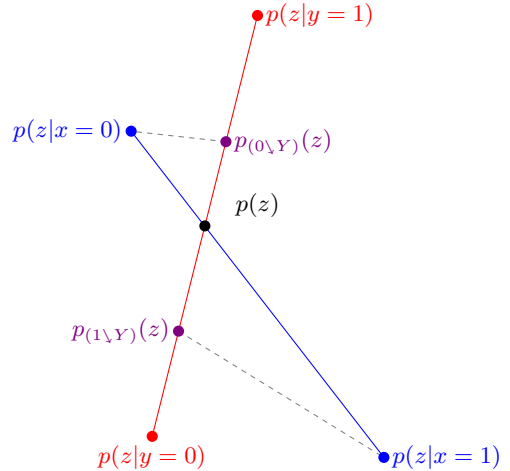
\begin{figure}[b]
  \begin{tikzpicture}[scale=1.4]
    \draw[red] (-0.5,-2) -- (0.5,2);
    \draw[blue] (-0.7,0.9) -- (1.7,-2.2);
    \draw[black!50, dash pattern=on 2pt off 2pt] (-0.25,-1) -- (1.7,-2.2);
    \draw[black!50, dash pattern=on 2pt off 2pt] (0.2,0.8) -- (-0.7,0.9);
    \fill[red!50!blue] (-0.25,-1)  circle(1.5pt) node[left] {$\qpr{1}{Y}(z)$};
    \fill[red!50!blue] (0.2,0.8)  circle(1.5pt) node[right] {$\qpr{0}{Y}(z)$};
    \fill[black] (0,0)  circle(1.5pt);
    \fill[red] (-0.5,-2)  circle(1.5pt) node[below] {$p(z|y=0)$};
    \fill[red] (0.5,2)  circle(1.5pt) node[right] {$p(z|y=1)$};
    \fill[blue] (-0.7,0.9)  circle(1.5pt) node[left] {$p(z|x=0)$};
    \fill[blue] (1.7,-2.2)  circle(1.5pt) node[right] {$p(z|x=1)$};
    \node[] at (0.5,0.2) {$p(z)$};
   \end{tikzpicture}
   \caption{\label{fig:projective_info} Construction of projective information for binary input variables.}
\end{figure}

\begin{lemma}
  Projected information $\Ipr{X}{Y}{Z}$ is well-defined, finite and non-negative.
\end{lemma}

\begin{proof}
  First, note that projected information can be written as the difference of two Kullback-Leibler divergences
  \begin{eqnarray} \Ipr{X}{Y}{Z} = \sum_x p(x) & &[ \DivKL{p(z|x)}{p(z)} \nonumber \\* &&  - \DivKL{p(z|x)}{\qpr{x}{Y}(z)}]. \nonumber\end{eqnarray}
  Therefore, if the projection is not unique, projected information only takes the KL-divergence into account which is the same for all possible solutions of the minimization problem in (\ref{eq:proj}).
  Now we have $\DivKL{p(z|x)}{\qpr{x}{Y}(z)} \leq \DivKL{p(z|x)}{p(z)}$ for all $x\in\mathcal{X}$ because of $p(z) \in \Ccl(\langle Y\rangle_Z)$ and the definition of $\qpr{x}{Y}(z)$ as the distance minimizing distribution to $p(\cdot|x)$ in $\Ccl(\langle Y\rangle_Z)$. Hence $\Ipr{X}{Y}{Z} \geq 0$. Furthermore $I(X;Z) = \sum_x p(x) \DivKL{p(z|x)}{p(z)} < \infty$.
\end{proof}

\subsubsection{Definition of Bivariate Redundancy}
The (bivariate) redundancy measure is now simply defined as the minimum of both projected information terms
\begin{equation} \Ired(Z;X,Y) := \min \{ \Ipr{X}{Y}{Z}, \Ipr{Y}{X}{Z}\}. \label{eq:redundancy_def}\end{equation}
At this point we can take the minimum over both values because we already corrected for the change of the distributions in different directions by projecting the conditionals. This is different to the approach taken by Williams and Beer \cite{Williams2010}, where the minimization does not consider that events in different source variables may change the distribution of the outcome in different directions in the geometrical space of distributions. Moreover, we define self-redundancy explicitly as
\begin{eqnarray}
  \Ired(Z;X) &:=&   \Ired(Z;X,X) \\
              &\,=&\Ipr{X}{X}{Z}.
\end{eqnarray}

\subsubsection{The Proposed Measure is a Bivariate Redundancy Measure}

To show that this is actually a redundancy measure, we have to show that it fulfils the four axioms (symmetry, self-redundancy, monotonicity and identity). Symmetry is obviously fulfilled, self-redundancy is also very quick to prove:
\begin{eqnarray}
  \Ired(Z;X) &=& \Ipr{X}{X}{Z} \\
    &=& \sum_{z,x} p(z,x) \log \frac{\qpr{x}{X}(z)}{p(z)} \\
    &=& \sum_{z,x} p(z,x) \log \frac{p(z|x)}{p(z)} \\
    &=& I(Z;X).
\end{eqnarray}
For the monotonicity axiom we first need to show $\Ired(Z;X,Y) \leq I(Z;X)$. Using the expression of projected information as a difference of Kullback-Leibler divergences we get
\begin{eqnarray}
  \Ired(Z;X,Y) &\leq& \Ipr{X}{Y}{Z} \\
    &=& \sum_x p(x) \left [ \DivKL{p(z|x)}{p(z)} \right.\nonumber \\*
    & & - \left. \DivKL{p(z|x)}{\qpr{x}{Y}(z)}\right] \\
    &=& I(Z;X) - \sum_x p(x) \DivKL{p(z|x)}{\qpr{x}{Y}(z)}. \nonumber
\end{eqnarray}
Hence it follows that $\Ired(Z;X,Y) \leq I(Z;X)$ as the KL-divergence is non-negative. To show equality holds if $X\subseteq Y$ we will first need the following two lemmas

\begin{widetext}
  \begin{lemma}\label{lem:proj1} For all $x \in \mathcal{X}$ and random variables $Y$ and $W$,
  \begin{equation} \sum p(z|x) \left(\log \qpr{x}{(Y,W)}(z)- \log \qpr{x}{Y}(z)\right) > 0. \end{equation}
\end{lemma}

\begin{proof}
  Let $x \in \mathcal{X}$, as $\Ccl(\langle Y \rangle_Z) \subseteq \Ccl (\langle(Y,W)\rangle_Z)$ (note that $p(y|z) = \sum_{w} p(y,w|z) $) we have due to the definition of the projection that
  \begin{eqnarray}
     \sum p(z|x) \log \frac{p(z|x)}{\qpr{x}{(Y,W)}(z)}  &\leq & \sum p(z|x) \log \frac{p(z|x)}{\qpr{x}{Y}(z)} \\
     \iff \sum p(z|x) \log \qpr{x}{(Y,W)}(z)  &\geq & \sum p(z|x) \log \qpr{x}{Y}(z)
  \end{eqnarray}

\end{proof}

\begin{lemma}\label{lem:proj2} For all $(y,w) \in \mathcal{Y}\times\mathcal{W}$
  \begin{equation} \sum p(z|y,w) \left(\log \qpr{(y,w)}{X}(z)- \log \qpr{y}{X}(z)\right) > 0. \end{equation}
\end{lemma}

\begin{proof}
  By definition, we have that $r=\qpr{(y,w)}{X}$ is minimizing $D_{KL}(p(z|y,w)||r)$ therefore
  \begin{eqnarray}
     \sum p(z|y,w) \log \frac{p(z|y,w)}{\qpr{(y,w)}{X}(z)}  &\leq & \sum p(z|y,w) \log \frac{p(z|y,w)}{\qpr{y}{X}(z)} \\
     \iff \sum p(z|y,w) \log \qpr{(y,w)}{X}(z)  &\geq & \sum p(z|y,w) \log \qpr{y}{X}(z)
  \end{eqnarray}

\end{proof}
\end{widetext}

Now the following proposition proves the missing piece for the monotonicity.
\begin{proposition}
  $\Ired(Z;X,Y) \leq \Ired(Z;X,(Y,W))$
\end{proposition}

\begin{proof}
From Lemma~\ref{lem:proj1} it follows directly that $\Ipr{X}{Y}{Z} \leq \Ipr{X}{(Y,W)}{Z}$, furthermore from Lemma~\ref{lem:proj2}, $\Ipr{Y}{X}{Z} \leq \Ipr{(Y,W)}{X}{Z}$ respectively. Hence, we conclude $\Ired(Z;X,Y) \leq \Ired(Z;X,(Y,W))$.
\end{proof}
Now it is only left to show that the measure also fulfils our new identity property, namely
\begin{equation} \Ired((X,Y);X,Y) = I(X;Y).\end{equation}
First we need the following lemma
\begin{lemma}
  If $Z=(X,Y)$ and $(x',y')$ denote an event of $Z$ then $\qpr{y'}{X}(x',y') = \qpr{x'}{Y}(x',y') = p(x'|y')p(y'|x') $
\end{lemma}

\begin{proof}
  Let $r\in \Ccl(\langle X \rangle_Z)$, it is of the form \begin{equation}r(x',y') = \sum_{x} \alpha_{x} p(x',y'|x) = \alpha_{x'} p(y'|x'),\end{equation} where $\alpha_{x}\geq0$ and $\sum \alpha_{x} = 1$. We also have
  \begin{eqnarray} \DivKL{p(\cdot|y)}{r} &=& \sum_{x',y'} p(x',y'|y) \log \frac{p(x',y'|y)}{\alpha_{x'} p(y'|x')} \\
  &=& \sum_{x'} p(x'|y) \log \frac{p(x'|y)}{\alpha_{x'} p(y|x')}\label{eq:min_prob}\end{eqnarray}
A simple calculation shows that the point $\alpha_{x'} = p(x'|y)$ fulfills the Karush-Kuhn-Tucker (KKT) conditions \cite{kuhntucker1950} for the minimization of Eq.~\eqref{eq:min_prob} with respect to the vector $\alpha_{x'}$ and the simplex constraints. The KL-divergence is convex in the second parameter and thus it follows from the KKT conditions that $\alpha_{x'} = p(x'|y)$ is a global solution for the constrained minimization of the KL-divergence $\DivKL{p(\cdot|y)}{r}$ parametrized by $\alpha_x$ as in Eq.~\eqref{eq:min_prob} and in turn $r(x',y')=p(x'|y)p(y|x')$. If we now set $y'=y$ then we get
$\qpr{y'}{X}(x',y') = p(x'|y')p(y'|x')$ and $\qpr{x'}{Y}(x',y') = p(x'|y')p(y'|x')$ respectively.
\end{proof}
And hence we can conclude our proof with the following proposition:
\begin{proposition}
  \label{prop:adjoin}
  $\Ipr{X}{Y}{X,Y} = \Ipr{Y}{X}{X,Y}  = I(X;Y)$
\end{proposition}

\begin{proof}
  Without loss of generality,
  \begin{eqnarray}
    & & \Ipr{X}{Y}{X,Y} \\
    &=& \sum_{x',y',x} p(x',y',x) \log \frac{\qpr{x}{Y}(x',y')}{p(x',y')} \\
    &=& H(X,Y) + \sum_{x',y'}  p(x',y') \log \qpr{x'}{Y}(x',y') \nonumber \\
    &=& H(X,Y) + \sum_{x,y} p(x,y) \log [p(x|y)p(y|x)] \nonumber \\
    &=& H(X,Y) - H(X|Y) - H(Y|X) \\
    &=& I(X;Y).
  \end{eqnarray}

\end{proof}
Thus $\Ired$ is a good candidate for measuring redundancy (in terms of redundancy with respect to some target variable).

\section{Comparisons}

Now that we have constructed a bivariate redundancy measure, we will present a few examples of redundancy calculations.
\subsection{Relation to Minimal Information}

There are some cases where $\Ired$ and $\Imin$ coincide and we will have a look at some of these cases later in Section~\ref{sec:examples}. In general there is a tendency of $\Imin$ to overestimate redundancy and in our examples it seems that $\Imin$ is an upper bound for $\Ired$ in most cases. There are a few exceptions, but it is not yet clear for which cases these exceptions appear or whether they are due to numerical instabilities. The overestimation of redundancy by $\Imin$ becomes predominant if the dimension of $Z$ is increased (see FIG.~\ref{fig:random_compare}). The explanation for this is that, the higher the dimension of the space gets, the larger the error becomes which results from not taking directionality into account.

\begin{figure}[b]
  \begin{tikzpicture}[scale=1.9]
    \bidecomp{$\{X\}\{Y\}$}{$\{X\}$}{$\{Y\}$}{$\{X,Y\}$}
  \end{tikzpicture}
  \caption{PI-diagram for the decomposition of the mutual information between $Z$ and $X,Y$ into PI-atoms. $\{X,Y\}$ denotes the synergistic, $\{X\}, \{Y\}$ the unique and $\{X\}\{Y\}$ the redundant part of the mutual information. }
  \label{fig:decomp_example}
\end{figure}
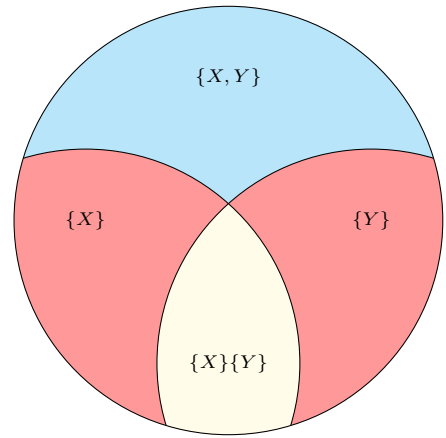

\begin{figure*}
    \centering
  \subfloat[$|\mathcal{Z}|=2$]{\label{fig:rc_z2}
      \begin{tikzpicture}
       \begin{axis}[
         scatterplot,
         xmin=0, xmax=0.5,
         ymin=0, ymax=0.5]
         \addplot[no marks, helpplot] {x};
         \addplot[only marks, dataplot] table[x index=7, y index=3] {data/random_z2_s.dat}; 
         \label{plot:inter_type}
       \end{axis}
    \end{tikzpicture}
  }\hskip1cm
  \subfloat[$|\mathcal{Z}|=4$]{\label{fig:rc_z4}
      \begin{tikzpicture}
       \begin{axis}[
         scatterplot,
         xmin=0, xmax=0.5,
         ymin=0, ymax=0.5]
         \addplot[no marks, helpplot] {x};
         \addplot[only marks, dataplot] table[x index=7, y index=3] {data/random_z4_s.dat}; 
       \end{axis}
    \end{tikzpicture}
  }\hskip1cm
  \subfloat[$|\mathcal{Z}|=6$]{\label{fig:rc_z6}
      \begin{tikzpicture}
       \begin{axis}[
         scatterplot,
         xmin=0, xmax=0.35,
         ymin=0, ymax=0.35]
         \addplot[no marks, helpplot] {x};
         \addplot[only marks, dataplot] table[x index=7, y index=3] {data/random_z6_s.dat}; 
       \end{axis}
    \end{tikzpicture}
  }\vskip0.75cm
  \subfloat[$|\mathcal{Z}|=8$]{\label{fig:rc_z8}
      \begin{tikzpicture}
       \begin{axis}[
         scatterplot,
         xmin=0, xmax=0.3,
         ymin=0, ymax=0.3]
         \addplot[no marks, helpplot] {x};
         \addplot[only marks, dataplot] table[x index=7, y index=3] {data/random_z8_s.dat}; 
       \end{axis}
    \end{tikzpicture}
  }\hskip1cm
  \subfloat[$|\mathcal{Z}|=20$]{\label{fig:rc_z20}
      \begin{tikzpicture}
       \begin{axis}[
         scatterplot,
         xmin=0, xmax=0.15,
         ymin=0, ymax=0.15]
         \addplot[no marks, helpplot] {x};
         \addplot[only marks, dataplot] table[x index=7, y index=3] {data/random_z20_s.dat}; 
       \end{axis}
    \end{tikzpicture}
  }\hskip1cm
  \subfloat[$|\mathcal{Z}|=40$]{\label{fig:rc_z40}
      \begin{tikzpicture}
       \begin{axis}[
         scatterplot,
         xmin=0, xmax=0.13,
         ymin=0, ymax=0.13]
         \addplot[no marks, helpplot] {x};
         \addplot[only marks, dataplot] table[x index=7, y index=3] {data/random_z40_s.dat}; 
       \end{axis}
    \end{tikzpicture}
  }
  \caption{Comparison of $\Imin$ and $\Ired$ for random distributions $p(x,y,z)$ and $|\mathcal{X}|=|\mathcal{Y}|=3$ and different sizes of $\mathcal{Z}$. Note that as the dimension goes up, $\Imin$ gets larger in comparison to $\Ired$. The distributions are initialized with random values, additionally the probability of each event being set to $0$ with probability $0.5$. }
  \label{fig:random_compare}
\end{figure*}
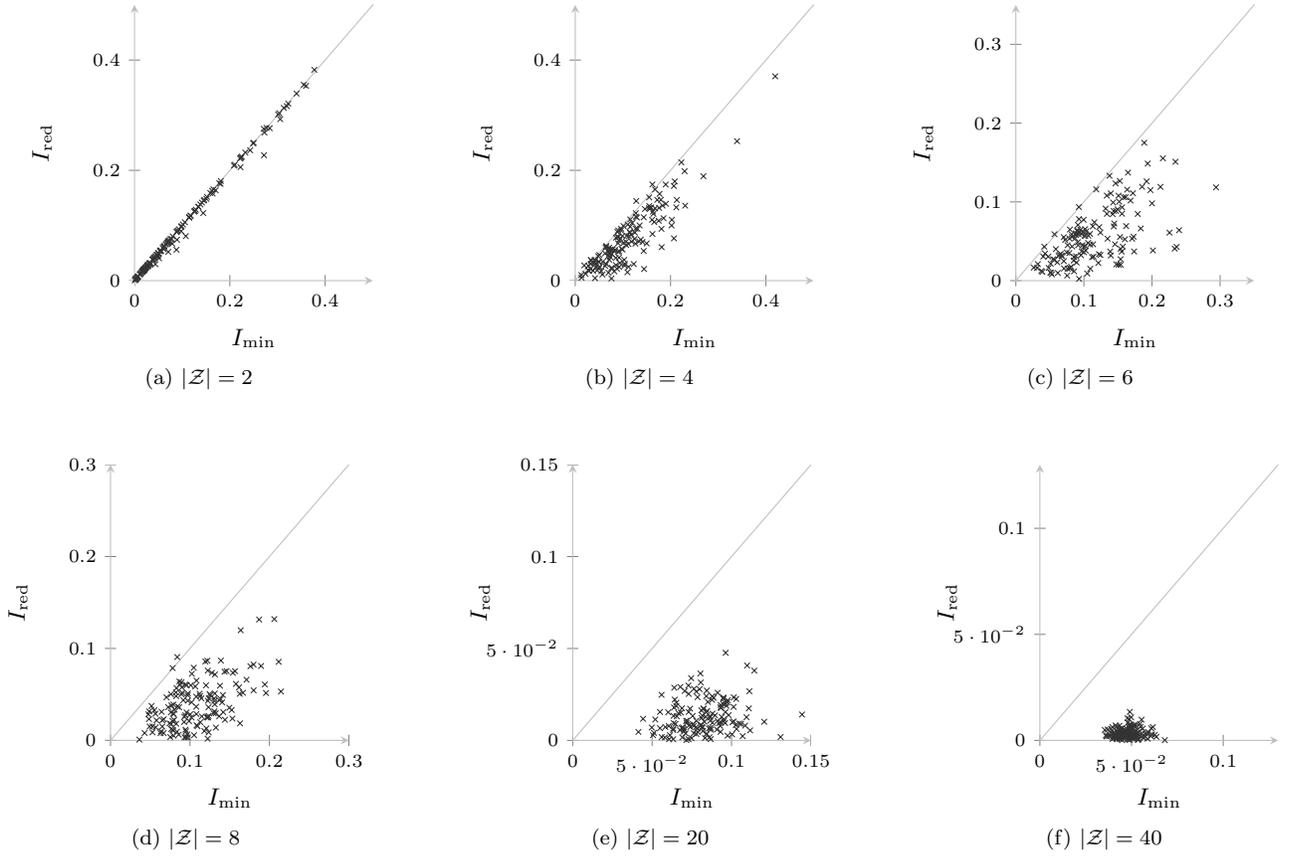

\subsection{Decomposition of Mutual Information}

In \cite{Williams2010} Williams and Beer introduce partial information atoms (PI-atoms) as a way to decompose multivariate mutual information into non-negative terms. These terms can be defined for any multivariate redundancy measure and denote redundant and synergistic contributions between several variables of a set of random variables $\mathbf{R}$ towards another random variable $Z$. They are denoted by $\piaBW{\mathbf{R}}(Z;\alpha)$ where $\alpha$ is a set of subsets of the base set of random variables $\mathbf{R}$. As this construction is possibly with any redundancy measure, we will use $\piaBW{\mathbf{R}}(Z;\alpha)$ denoting the PI-atoms based on $\Imin$ as a redundancy measure and thereby staying consistent in the notation with \cite{Williams2010}. The primed version $\pia{\mathbf{R}}(Z;\alpha)$ on the other hand will denote the decomposition using the redundancy measure $\Ired$ introduced here.

In the bivariate case, this leads to the decomposition of mutual information $I(Z;X,Y)$ into four partial information atoms. Here we have $\mathbf{R}=\{X,Y\}$. Now, following \cite{Williams2010} there are four atomic terms,
\begin{itemize}
   \item $\pia{\mathbf{R}}(Z;\{X\}\{Y\})$ which is the redundant information contained in $X$ and $Y$ about $Z$,
   \item $\pia{\mathbf{R}}(Z;\{X\})$ and $\pia{\mathbf{R}}(Z;\{Y\})$ are the unique information about $Z$, which is only contained in $X$ or $Y$ respectively,
   \item and $\pia{\mathbf{R}}(Z;\{X,Y\})$, synergistic information, the information about $Z$ that is only available if $X$ and $Y$ are both known.
 \end{itemize}
The sum of these terms is exactly the mutual information between $Z$ and all sources, i.e. \begin{eqnarray}I(Z;X,Y)&=&\pia{\mathbf{R}}(Z;\{X\}\{Y\})+\pia{\mathbf{R}}(Z;\{X\})\nonumber \\ & & + \pia{\mathbf{R}}(Z;\{Y\})+\pia{\mathbf{R}}(Z;\{X,Y\}). \label{eq:pia_decomp}\end{eqnarray} as well as
\begin{equation}I(Z;X)=\pia{\mathbf{R}}(Z;\{X\}\{Y\})+\pia{\mathbf{R}}(Z;\{X\}) \label{eq:pia_mono_decomp}\end{equation}
and for $Y$ respectively. Still following \cite{Williams2010}, but having replaced $\Imin$ by $\Ired$ we get $\pia{\mathbf{R}}(Z;\{X\}\{Y\})=\Ired(Z;X,Y)$ and $\pia{\mathbf{R}}(Z;\{X\})=I(Z;X)-\Ired(Z;X,Y)$. Finally, for the synergistic term \begin{eqnarray}\pia{\mathbf{R}}(Z;\{X,Y\})&=&I(Z;X,Y)-\pia{\mathbf{R}}(Z;\{X\}) \nonumber \\* & & -\pia{\mathbf{R}}(Z;\{Y\})  \nonumber \\* & &  -\pia{\mathbf{R}}(Z;\{X\}\{Y\})\\&=&I(Z;X,Y)-I(Z;X)  \nonumber \\* & &  -I(Z;Y)+\Ired(Z;X,Y). \label{eq:synergy_def}\end{eqnarray}
Now this decomposition is not non-negative by default and this needs to be shown for the specific redundancy measure used. It is shown by Williams in \cite{Williams2011a} for the decomposition using $\Imin$. Here, we will show it for the bivariate case with $\Ired$ as redundancy measure: Firstly, $\Ired(Z;X,Y)$ is non-negative, as shown earlier, furthermore it follows from the axioms of the redundancy measure that $\Ired(Z;X,Y)\leq I(X;Z)$ and with the same argument $\Ired(Z;X,Y)\leq I(Y;Z)$ which immediately implies that the unique information terms are non-negative. The following lemma now gives the non-negativity of the synergistic term:

\begin{widetext}
\begin{lemma}
  $I(Z;X,Y)-I(Z;X)-I(Z;Y)+\Ipr{X}{Y}{Z} \geq 0$
\end{lemma}
\begin{proof}
  We can reformulate the left hand side
  \begin{eqnarray}
    & & I(Z;X,Y)-I(Z;X) -I(Z;Y)+\Ipr{X}{Y}{Z} \\
    &=& I(Z;X,Y)-I(Z;Y) - \sum_x p(x) \DivKL{p(z|x)}{\qpr{x}{Y}(z)} \\
    &=& \sum_{x,y} p(x,y) \DivKL{p(z|x,y)}{p(z|y)} - \sum_x p(x) \DivKL{p(z|x)}{\qpr{x}{Y}(z)}\\
    &=& \sum_{x} p(x) \left( \left(\sum_{y} p(y|x) \DivKL{p(z|x,y)}{p(z|y)} \right)- \DivKL{p(z|x)}{\qpr{x}{Y}(z)} \right) \label{eqn:before_convex}
  \end{eqnarray}
  and now by the convexity of the Kullback-Leibler divergence:
  \begin{eqnarray}
    &\geq& \sum_{x} p(x) \left( \DivKL{\sum_{y} p(y|x) p(z|x,y)}{\sum_{y} p(y|x)p(z|y)}- \DivKL{p(z|x)}{\qpr{x}{Y}(z)} \right)\\
    &=& \sum_{x} p(x) \left( \DivKL{p(z|x)}{r(z|x)}- \DivKL{p(z|x)}{\qpr{x}{Y}(z)} \right)
  \end{eqnarray}
  where $r(z|x):=\sum_{y} p(y|x)p(z|y) \in \Ccl(\langle Y \rangle_Z)$ and thus
  \begin{equation}\DivKL{p(z|x)}{r(z|x)}- \DivKL{p(z|x)}{\qpr{x}{Y}(z)} \geq 0 \mbox{ for all } x \in \mathcal{X}.\end{equation}
\end{proof}
\end{widetext}

Given the non-negativity of the decomposition, we can visualize it using a PI-diagram as seen in FIG.~\ref{fig:decomp_example}. The whole circle represents the mutual information $I(Z;X,Y)$ and the colored/shaded regions represent redundant (yellow/light shaded), unique (red/dark shaded) and synergistic (blue/medium shaded) information.

\subsection{Examples} \label{sec:examples}

We will now go through some examples for the bivariate measure, in particular those discussed in \cite{Griffith2011}, which are a good selection of test cases for the desired properties of a redundancy/synergy measure.

\subsubsection{Copying - From Redundancy to Uniqueness}
\label{sec:ex_copy}
\begin{figure}[b]
\subfloat[Bayesian model]{\label{fig:lambdamodel}
  \begin{tikzpicture}[scale=1.2]
  \GraphInit[vstyle=Dijkstra]
  \SetVertexMath \renewcommand*{\VertexLineColor}{white}
  \Vertex[L=W,Lpos=90]{W}
  \NOEA[L=X,Lpos=90](W){X}
  \SOEA[L=Y,Lpos=-90](W){Y}
  \SOEA[L=Z,Lpos=-90](X){Z}
  \SetUpEdge[style={post}]
  \tikzstyle{EdgeStyle}=[line width=.5pt]
  \tikzstyle{LabelStyle}=[left=10pt]
  \tikzstyle{LabelStyle}=[above=3pt]
  \Edge[label=\scriptsize$\lambda$,labelstyle={left=8pt,below=4pt}](W)(X)
  \Edge[label=\scriptsize$\lambda$,labelstyle={left=7pt,below=18pt}](W)(Y)
  \Edge(X)(Z)
  \Edge(Y)(Z)
\end{tikzpicture}
}\hskip0.5cm
\subfloat[PI-diagram for $\lambda=0$, complete redundancy (\textsc{Rdn})]{\label{fig:lambda0}
  \begin{tikzpicture}[scale=0.5]
    \bidecomphigh{1}{0}{0}{0}{70}{20}{20}{10}
  \end{tikzpicture}}\hskip0.5cm
   \subfloat[PI-diagram for $\lambda=1$, complete uniqueness (\textsc{Unq})]{\label{fig:lambda1}
  \begin{tikzpicture}[scale=0.5]
    \bidecomphigh{0}{1}{1}{0}{5}{70}{70}{10}
  \end{tikzpicture}}
  \caption{Copy Example. Complete redundancy and complete uniqueness using $\Ired$.}
  \label{fig:lambda_compare}
\end{figure}
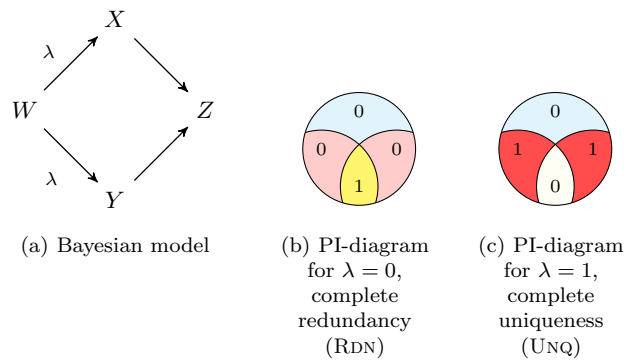
Our first example is a very simple mechanism which simply copies the binary input variables $X$ and $Y$ into $Z$, i.e. $Z=(X,Y)$. However, we also add a control paremeter $\lambda\in[0,1]$ which determines how correlated $X$ and $Y$ are, as follows: Let $W$ be a uniformly distributed binary random variable, $p(x|w) = \lambda \frac{1}{2} + (1-\lambda)\delta_{xw}$ and $p(y|w) = \lambda \frac{1}{2} + (1-\lambda)\delta_{yw}$. For $\lambda=1$ we have that $X$ and $Y$ are independent and we recover the example ``\textsc{Unq} (Unique Information)'' from \cite{Griffith2011}. On the other extreme $\lambda=0$ we have that $X$ and $Y$ are identical copies of $W$ and therefore $Z$ is equivalent to $W$ from an information theoretic point of view. This is also reflected in the decomposition as in this case $I(Z;X,Y)=I(W;X,Y)$ and $\Ired(Z;X,Y)=\Ired(W;X,Y)$, so we can see that this is the example ``\textsc{Rdn} (Redundant Information)'' from \cite{Griffith2011}. By varying $\lambda$ we can vary the entropy of the outcome $Z$ and at the same time exchange unique information for redundancy. FIG.~\ref{fig:lambda_compare} illustrates the decomposition at both extremal values of $\lambda$ and it can be seen that the resulting values of $\Ired$ coincide with the proposed values in \cite{Griffith2011}. The effect of changing $\lambda$ is shown in FIG.~\ref{fig:lambda_plot}.

\begin{figure}[t]
    \centering
\begin{tikzpicture}
   \begin{axis}[
     xmin=0, xmax=1,
     ymin=0, ymax=2,
     articleplot,
      xlabel={$\lambda$},
      ylabel={bits},
      width=8cm,
      height=6cm]
     \addplot[no marks, line width=1pt, darkgray] table[x index=0, y index=4] {data/lambda.dat}; \label{plot:lambda_red} 
     \addplot[no marks, gray, dash pattern=on 6pt off 2pt] table[x index=0, y index=3] {data/lambda.dat}; \label{plot:lambda_mi}
     \addplot[no marks, line width=1pt, darkgray, dash pattern=on 2pt off 2pt] table[x index=0, y index=8] {data/lambda.dat}; \label{plot:lambda_imin}
   \end{axis}
\end{tikzpicture}
  \caption{Comparison of total mutual information I(Z;X,Y) (\ref{plot:lambda_mi}), our redundancy measure $\Ired$ (\ref{plot:lambda_red}) and $\Imin$ (\ref{plot:lambda_imin}) for varying values of $\lambda$, where $\lambda$ controls the correllation between $X$ and $Y$. It can be seen $\Imin$ measures a constant amount of redundancy and therefore does not distinguish between redundancy and uniqueness with varying $\lambda$ as desired, whereas $\Ired$ does.}
  \label{fig:lambda_plot}
\end{figure}
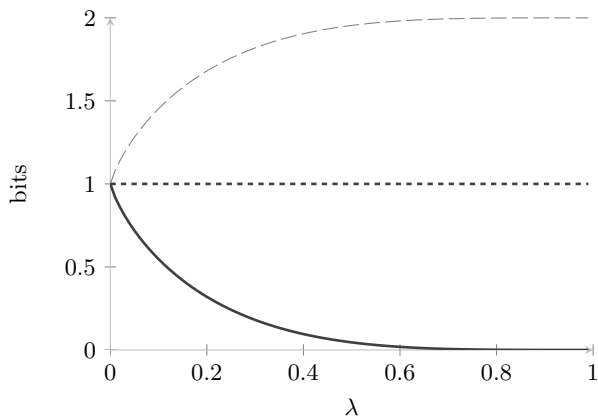

\subsubsection{XOR}

The \textsc{Xor} gate ($\oplus$), is a classical example for the appearance of synergy, in the sense of the whole being more than the sum of the individuals.  We expect to only observe synergistic information, as the result is only known if both inputs are available, and the uncertainty given one input is the same as giving no input at all. Again the inputs are uniformly distributed binary random variables and $Z=X \oplus Y$. In fact, in this case we have $\Ired(Z;X,Y)=\Imin(Z;X,Y)=0$ and get the purely synergistic decomposition as illustrated in FIG.~\ref{fig:xor_example}. Note that $\Ired$ defines the redundancy, other terms are all derived by the decomposition.

\subsubsection{AND - Mechanisms at Work}

\label{sec:and_example}
We now come to the \textsc{And} gate, $Z=X\land Y$. This turns out to be an interesting case, because it demonstrates the subtle difference between redundant information that is due to the ``ignorance'' of the mechanism with respect to the source, and redundancy that is already apparent in the sources. In \cite{Griffith2011, Griffith2012} it is argued that vanishing mutual information between the sources $X$ and $Y$ themselves implies vanishing redundant information\footnote{``However, because $X_1$ and $X_2$ are independent, [...], thus necessitating there is zero redundant information [...].'',\cite{Griffith2011}}. This feature is also shared by the synergy measure introduced in \cite{Griffith2012}. However, here we would like to embrace a different view on redundant information: even if the sources are independent, there can be a correlation in the change of the distribution over $Z$ given observations in $X$ and $Y$ respectively. Observing one input does not give any information about the other input, but part of the information gain about the distribution of the output can be the same as one gets from the other input alone. In particular in the case of the \textsc{And} gate, observing a 0 in either input leads to $p(z=0)=1$. As a result of calculating the redundancy for this example we get $\Ired(Z;X,Y)=\Imin(Z;X,Y)=0.311278$, so this is another example where minimal and redundant information coincide. FIG.~\ref{fig:and_example} illustrates the decomposition of the total mutual information for this example.

We denote redundant information that is only due to the mechanism, as it is the case here, \textit{mechanistic redundancy}. Contrary to this we call redundant information that already appears in the inputs \textit{source redundancy}. Redundancy in the source must already manifest itself in the mutual information between the inputs. We do not give a rigorous definition for these terms, as it can be seen in the next example, there are cases where it is not clear how to separate both. However, if there is positive redundant information $\Ired > 0$ but vanishing mutual information between the sources, we will attribute all redundant information to mechanistic redundancy.

\begin{figure}
\subfloat[PI-diagram]{\label{fig:xor_diag}
  \begin{tikzpicture}[scale=0.8]
    \bidecomphigh{0}{0}{0}{1}{5}{20}{20}{50}
  \end{tikzpicture}}\hskip0.25cm
   \subfloat[circuit diagram]{\label{fig:cor_circuit}
    \begin{tikzpicture}[circuit logic US, minimum height=10mm,ampersand replacement=\&]
\matrix[column sep=10mm]
{
\node (x) {$X$}; \& \&  \\
\& \node [xor gate] (xor1) {\scriptsize\,\,XOR}; \& \node (z) {$Z$};\\
\node (y) {$Y$}; \& \& \\
};
\draw[*-] (x.east) -- ++(right:3mm) |- (xor1.input 1);
\draw[*-] (y.east) -- ++(right:3mm) |- (xor1.input 2);
\draw[-*] (xor1.output) -- ++(right:3mm) |- (z.west);
\end{tikzpicture}}
\vskip0.125cm
  \caption{\textsc{Xor} Example. A purely synergistic mechanism.}
  \label{fig:xor_example}
\end{figure}
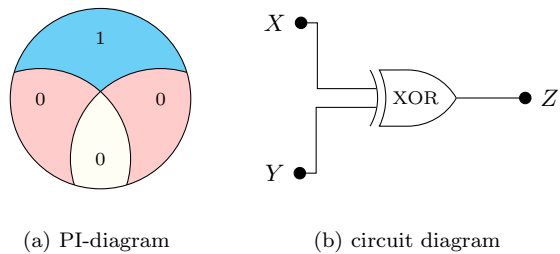
\begin{figure}[t]
\subfloat[PI-diagram]{\label{fig:and_diag}
  \begin{tikzpicture}[scale=0.8]
    \bidecomphigh{0.311}{0}{0}{0.5}{50}{20}{20}{50}
  \end{tikzpicture}}\hskip0.25cm
   \subfloat[circuit diagram]{\label{fig:cor_andircuit}
    \begin{tikzpicture}[circuit logic US, minimum height=10mm,ampersand replacement=\&]
\matrix[column sep=10mm]
{
\node (x) {$X$}; \& \&  \\
\& \node [and gate] (and1) {\scriptsize AND}; \& \node (z) {$Z$};\\
\node (y) {$Y$}; \& \& \\
};
\draw[*-] (x.east) -- ++(right:3mm) |- (and1.input 1);
\draw[*-] (y.east) -- ++(right:3mm) |- (and1.input 2);
\draw[-*] (and1.output) -- ++(right:3mm) |- (z.west);
\end{tikzpicture}}
\vskip0.12cm
  \caption{\textsc{And} Example. The total mutual information is $I(Z;X,Y)=0.811278$.}
  \label{fig:and_example}
\end{figure}
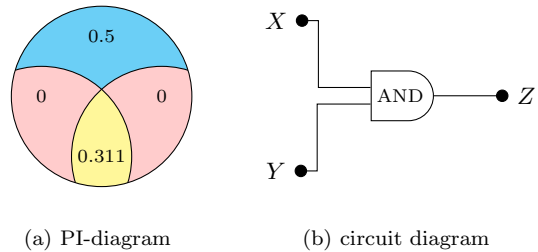

\subsubsection{Summing Dice}
Let us now consider an example where we throw two dice (cubic dice, with numbered sides from 0 to 5), represented by the random variables $D_1$, $D_2$ and sum their results. There are several ways to sum the results, we could simply add the two results --- this would lead to results ranging from 0 to 10 where 5 is the most probable result and 0 or 10 the least probable results --- or we multiply the result of the first die by 6 to get a uniform distribution of all numbers ranging from 0 to 35. Indeed, we will also look at all intermediate summations defined by $R=\alpha D_1 + D_2$ where $\alpha \in \{1,2,3,4,5,6\}$. Our hypothesis was that for the direct summation ($\alpha=1$) there is a positive amount of redundancy between $D_1$ and $D_2$ with respect to $R$, because knowing the roll of one die gives ``overlapping'' information (in the same direction in the space of distributions) with the roll of the other die about the final result. The redundancy should then decrease if $\alpha$ is increased, up to the point where $\alpha=6$ and the sum of both dice rolls is isomorphic to the joint variable of the two dice rolls, i.e. $6D_1 + D_2 \simeq (D_1,D_2)$. Indeed, this is reflected in the redundancy $\Ired(R;D_1,D_2)$. In FIG.~\ref{fig:dice} we added an additional parameter $\lambda$ that controls how correlated the two dice are, in the same way as $\lambda$ was introduced in the copy example in Section~\ref{sec:ex_copy} to control the correlation between the input variables. For $\lambda=1$ they are independent and it can be seen that the redundancy increases with decreasing $\alpha$, on the other extreme $\lambda=0$ the dice are completely correlated. In this case we can see that the redundancy is already existent in the source ($I(D_1,D_2)\approx 2.58$) shadows all redundancy otherwise induced through the mechanism and hence there is no difference in the redundancy value for all values of $\alpha$.

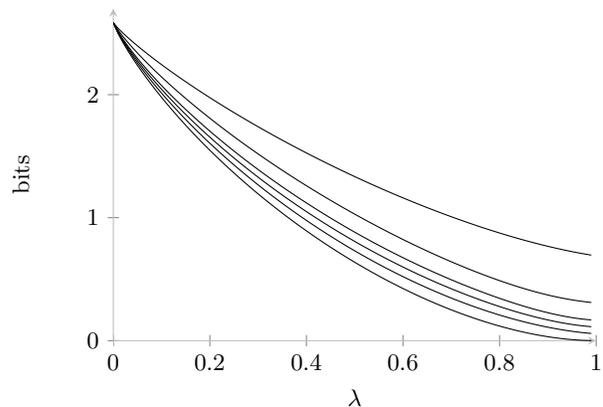
\begin{figure}
  \centering
\begin{tikzpicture}
   \begin{axis}[
     articleplot,
      xlabel={$\lambda$},
      ylabel={bits},
      width=8cm,
      height=6cm,
      xmin=0, xmax=1,
      ymin=0, ymax=2.7,
      no marks]
      \addplot[black] table[x index=0, y index=5] {data/dice1.dat}; \label{plot:dice1}
      \addplot[black] table[x index=0, y index=5] {data/dice2.dat}; \label{plot:dice2}
      \addplot[black] table[x index=0, y index=5] {data/dice3.dat}; \label{plot:dice3}
      \addplot[black] table[x index=0, y index=5] {data/dice4.dat}; \label{plot:dice4}
      \addplot[black] table[x index=0, y index=5] {data/dice5.dat}; \label{plot:dice5}
      \addplot[black] table[x index=0, y index=5] {data/dice6.dat}; \label{plot:dice6}
   \end{axis}
\end{tikzpicture}
  \caption{Plot of the redundant information $\Ired(R;D_1,D_2)$ depending on the correlation $\lambda$ between the two dice $D_1$ and $D_2$. From top to bottom the summation coefficient is $\alpha = 1,...,6$. It can be seen that for independent dice $\lambda=1$ the amount of redundancy depends on the mechanism that is used to sum the results, whereas on the other extreme, all redundancy comes from the correlation of the sources.}
  \label{fig:dice}
\end{figure}

\subsubsection{Composition of Mechanisms}

The last three examples from \cite{Griffith2011} are compositions of the already shown examples. The first one \textsc{RdnXor} combines the redundant copy example ($\lambda=0$) with an \textsc{Xor} gate: $(X,W)$ and $(Y,W)$ are the inputs and $Z=(W,X\oplus Y)$ is the output. With our redundancy measure, this results in the required composite of one bit of redundant and one bit of synergistic information, the same as measured with $\Imin$.

The second example \textsc{RdnUnqXor}, combines an \textsc{Xor} gate with the two extremal copy cases. The inputs are $(X_1,X_2,W)$ and $(Y_1,Y_2,W)$, all independent and uniformly distributed. The output is $Z=(X_1 \oplus Y_1, (X_2,Y_2),W)$. Here we get the intended 1 bit of information in every partial information term, i.e. 1 bit of redundant, 1 bit synergistic information and 1 bit unique information per input, and a total 4 bits of mutual information.

The third example \textsc{XorAnd}, combines an \textsc{Xor} gate with an \textsc{And} gate, i.e. $Z=(X\land Y, X\oplus Y)$. This obviously leads to a different result than in \cite{Griffith2011}, as the same effect of mechanistic redundancy appears in the \textsc{And} gate, as mentioned in Section~\ref{sec:and_example}.

\subsubsection{Summary}

\begin{table}[b]
\begin{ruledtabular}

  \begin{tabular}{@{}l@{\hspace{1cm}}lll@{}}
Example & Expected & $\Ired$ & $\Imin$ \\
\addlinespace[5pt]
Copy ($\lambda=0$) / \textsc{Rdn} & 1 & 1 & 1 \\
Copy ($\lambda=1$) / \textsc{Unq} & 0 & 0 & \color{red} 1 \\
\textsc{Xor} & 0 & 0 & 0 \\
\textsc{And} & 0.311 & 0.311 & 0.311 \\
\addlinespace[5pt]
\textsc{RdnXor} & 1 & 1 & 1 \\
\textsc{RdnUnqXor} & 1 & 1 & \color{red} 2 \\
\textsc{XorAnd} & 0.5 & 0.5 & 0.5 \\
\addlinespace[5pt]
Copy ($\lambda<1$) & I(X;Y) & I(X;Y) & \color{red} 1 \\
\end{tabular}
\end{ruledtabular}
  \caption{Summary of the bivariate redundancy examples. Results for the calculations of the examples using $\Ired$ and $\Imin$, as well as the expected value that results from considarations of the desired properties of a redundancy measure, cf. \cite{Griffith2011}. }
  \label{tab:decomp}
\end{table}

In summary, these examples show that $\Ired$ captures proposed the concept of redundancy very well. Furthermore the resulting decomposition is in agreement with the desired examples in \cite{Griffith2011} except for the case where what we call mechanistic redundancy appears, which was not accounted for in the comparison of current measures of synergy. TABLE~\ref{tab:decomp} summarises the comparison of $\Imin$ and $\Ired$.

\subsection{Information Transfer}

In \cite{Williams2011} the partial information decomposition is used to introduce new measures of information transfer. The measures are based on a decomposition of transfer entropy. Transfer entropy, introduced by Schreiber \cite{schreiber2000}, is defined for two random processes $X_t$ and $Y_t$ as
\begin{equation}
  T_{Y\to X} = I(X_{t+1};Y_{t}|X_{t}).
\end{equation}
It measures the influence of the process $Y$ at time t on the state of the process $X$ in the next time step. One can also take a longer history instead of $Y_{t}$ and $X_{t}$ into account. Conditional mutual information is defined as
\begin{equation}
  I(X_{t+1};Y_{t}|X_{t}) = I(X_{t+1};Y_{t},X_{t}) - I(X_{t+1};X_{t}).
\end{equation}
As the conditional entropy is the difference of two mutual information terms, the PI-decomposition can be used to decompose transfer entropy into two non-negative components. The decomposition is illustrated in FIG.~\ref{fig:te_decomp}. Let $\mathbf{R}=\{X_{t},Y_{t}\}$ then it follows from (\ref{eq:pia_decomp}) and (\ref{eq:pia_mono_decomp}) that
\begin{equation}
  T_{Y\to X} = \pia{\mathbf{R}}(X_{t+1};\{Y_t\}) + \pia{\mathbf{R}}(X_{t+1};\{X_t,Y_t\}).
\end{equation}
The first term denotes all information that uniquely comes from $Y_t$, called {\em State Independent Transfer Entropy} (SITE) by Williams and Beer \cite{Williams2011}. The second term on the other hand denotes information that comes from $Y_t$ but depends on the state of $X_t$ and thus is called {\em State Dependent Transfer Entropy} (SDTE) in \cite{Williams2011}. We now apply both measures $\Imin$ (with corresponding PI-atoms $\piaBW{\mathbf{R}}$) and $\Ired$ (with corresponding PI-atoms $\pia{\mathbf{R}}$) as the underlying redundancy measure for the decomposition and compare the results.

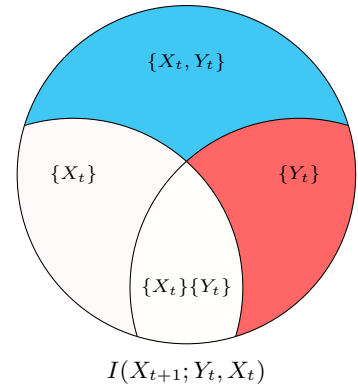
\begin{figure}[t!]
  \begin{tikzpicture}[scale=1.5]
    \bidecomphigh{$\{X_t\}\{Y_t\}$}{$\{X_t\}$}{$\{Y_t\}$}{$\{X_t,Y_t\}$}{2}{2}{60}{60}
    \node at (0,-1.75cm) {$I(X_{t+1};Y_{t},X_{t})$};
  \end{tikzpicture}
  \caption{PI-diagram for the decomposition of transfer entropy into PI-atoms. The coloured areas denote the transfer entropy.}
  \label{fig:te_decomp}
\end{figure}

We will consider two examples to show the difference of the decomposition when using $\Ired$ instead of $\Imin$. The first one revisits an example from \cite{Williams2011} where $X$ and $Y$ are two binary, coupled Markov random processes. The process $Y$ is uniformly i.i.d. and $x_{t+1}=y_t$ if $x_t=0$, moreover 
\begin{eqnarray}
  p(x_{t+1}=y_t|x_t=1)&=&1-d, \\
  p(x_{t+1}=1-y_t|x_t=1)&=&d.
\end{eqnarray} So $d\in [0,1]$ controls whether there is any dependence on the previous state of $X$. If $d$ vanishes $X$ is simply a copy of $Y$. For this example and $d=0$ shows only state-independent transfer while $d=1$ shows only state dependent transfer and most importantly the decompositions of transfer entropy using either measure ($\Ired,\Imin$) coincide (compare with FIG.~\ref{fig:binary_markov}).

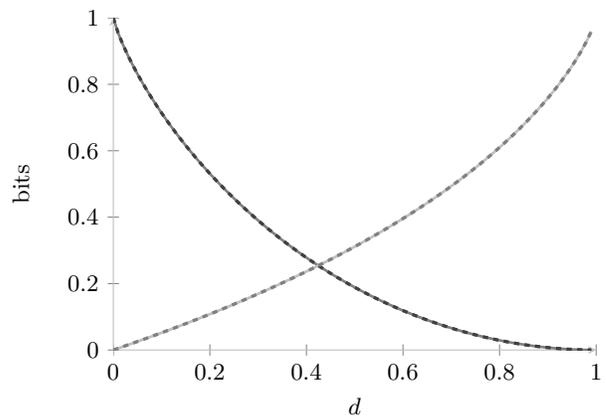
\begin{figure}
  \centering
\begin{tikzpicture}
   \begin{axis}[
     articleplot,
      xlabel={$d$},
      ylabel={bits},
      width=8cm,
      height=6cm,
      xmin=0, xmax=1,
      ymin=0, ymax=1]
      \addplot[no marks, line width=1pt, gray] table[x index=0, y index=6] {data/binary_markov.dat}; \label{plot:bm_site_r}
      \addplot[no marks, line width=1pt, lightgray] table[x index=0, y index=7] {data/binary_markov.dat}; \label{plot:bm_sdte_r}
      \addplot[no marks, line width=1.25pt, darkgray,  dash pattern=on 2pt off 2pt] table[x index=0, y index=10] {data/binary_markov.dat}; \label{plot:bm_site}
      \addplot[no marks, line width=1.25pt, gray,  dash pattern=on 2pt off 2pt] table[x index=0, y index=11] {data/binary_markov.dat}; \label{plot:bm_sdte}

   \end{axis}
\end{tikzpicture}
  \caption{Decomposition of transfer entropy $T_{Y \to X}$ for the first example process. The plot shows SITE (\ref{plot:bm_site} using $\Imin$, \ref{plot:bm_site_r} using $\Ired$) and SDTE (\ref{plot:bm_sdte} using $\Imin$, \ref{plot:bm_sdte_r} using $\Ired$) given $d$. It can be seen that both decompositions coincide for this process.}
  \label{fig:binary_markov}
\end{figure}

The second example, though constructed for this specific purpose, is more intricate. First of all it shows the difference between the two measures, but it is also a good example of the subtlety of redundancy in mechanisms. Let us consider the following two processes $(X_t,Y_t)$ and $Z_t$ where $Z_t$ are uniformly i.i.d. random variables, $X_{t+1}$ is a copy of $X_t$ and \begin{equation}p(y_{t+1}|y_t,z_t) = (1-d) \delta_{y_t y_{t+1}} + d \delta_{z_t y_{t+1}}.\end{equation}
The process $Y_t$, copies with probability $d$ the value of $Z_{t-1}$ and with probability $(1-d)$ the value of $Y_{t-1}$. We now measure the transfer entropy $T_{Z \to (X,Y)}$, see FIG.~\ref{fig:proc2_bayesian} for a Bayesian network of the process.

\begin{figure}

  \begin{tikzpicture}[scale=1.8]
  \GraphInit[vstyle=Dijkstra]
  \SetVertexMath \renewcommand*{\VertexLineColor}{white}
  \Vertex[L=X_t,Lpos=90]{Xt}
  \SO[L=X_{t+1},Lpos=90](Xt){Xtp}
  \EA[L=Y_t,Lpos=-90](Xt){Yt}
  \SO[L=Y_{t+1},Lpos=90](Yt){Ytp}
  \EA[L=Z_t,Lpos=-90](Yt){Zt}
  \SO[L=Z_{t+1},Lpos=90](Zt){Ztp}
  \SetUpEdge[style={post}]
  \tikzstyle{EdgeStyle}=[line width=.5pt]
  \tikzstyle{LabelStyle}=[left=10pt]
  \tikzstyle{LabelStyle}=[above=3pt]
  \Edge[label=\scriptsize$(1-d)$,labelstyle={left=18pt,below=10pt}](Yt)(Ytp)
  \Edge[label=\scriptsize$d$,labelstyle={right=16pt,below=7pt}](Zt)(Ytp)
  \Edge(Xt)(Xtp)
\end{tikzpicture}
  \caption{Bayesian network of the second example process. $X_t$ is a parallel and independent process, the only information transfer between the processes is from $Z_t$ to $Y_{t+1}$. }
  \label{fig:proc2_bayesian}
\end{figure}
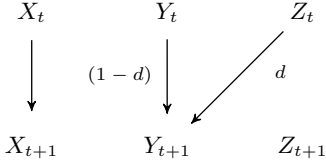

It can be seen in FIG.~\ref{fig:binary_markov2} that the two decompositions coincide for $d\leq 0.5$. For $d=0$ the two processes are completely independent which is reflected in the vanishing overall transfer entropy in this case. On the other extreme using $d=1$, the decomposition using $\Ired$ gives complete state-independent transfer entropy while the decomposition using $\Imin$ sees total state-dependent transfer entropy. In this case the decompositions disagree completely and we argue that our measure reflects the process much better. With $d=1$ the process always copies $Z_t$ to $Y_{t+1}$, which is completely independent of $(X_t,Y_t)$. Specifically, $\Imin$ mistakenly sees redundancy between $X_t$ and $Z_t$ in the evolution of one timestep. Following (\ref{eq:synergy_def}) and (\ref{eq:pia_mono_decomp}) this is then reflected in the vanishing state-independent transfer entropy for all $d$ (larger redundancy means more synergy and less unique information, given that the mutual information stays constant).

The fact that $\Imin$ measures more redundancy has the same reason why $\Imin$ measures redundancy between independent $X$ and $Y$ with respect to $Z=(X,Y)$, namely it compares changes in different direction in the space of distributions. The parallel and independent process $X_t$ lets $\Imin$ see a dependency between the two processes $X_t$ and $Z_t$ that does not exist. If we consider the transfer entropy $T_{Z \to Y}$ from $Z_t$ to $Y_t$ only, ignoring the process $X_t$ completely, we can see in FIG.~\ref{fig:binary_markov3} that the decomposition (\ref{plot:bm_site3},\ref{plot:bm_sdte3}) now coincides with the decomposition of $T_{Z \to (X,Y)}$ using $\Ired$ (\ref{plot:bm_site_r2},\ref{plot:bm_sdte_r2} in FIG. \ref{fig:binary_markov2}).

Nonetheless, we have not yet explained the quite unusual non-differentiable shape of the state-independent transfer entropy, which only is positive for $d>0.5$. This is surprising because up to $d=0.5$ all transfer entropy is considered to be state-dependent, even though with probability $d$ the state of $Y_{t+1}$ takes on the state of $Z_t$. As the process $X_t$ was only used to demonstrate that using $\Imin$ for the decomposition measures state dependencies in the transfer-entropy that are not there, we will now leave $X_t$ aside and only consider the process $(Y_t,Z_t)$ as described above.

To understand the shape of the graph of state-dependent transfer entropy of this process, we need to have a look at the mutual information $I((Y_{t+1});Z_t)$ (\ref{plot:bm_mi} in FIG.~\ref{fig:binary_markov2b}) and the redundancy $\Ired(Y_{t+1};Y_t,Z_t)$ (\ref{plot:bm_red} in FIG.~\ref{fig:binary_markov2b}). From (\ref{eq:pia_mono_decomp}) it follows that the state-independent transfer entropy (\ref{plot:bm_site_r2} in FIG.~\ref{fig:binary_markov2} and \ref{plot:bm_site3} in FIG.~\ref{fig:binary_markov3}) is now the difference of these two terms (compare with FIG.~\ref{fig:te_decomp}).

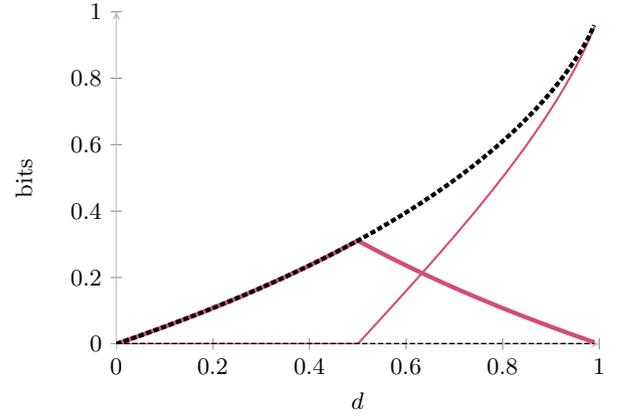
\begin{figure}
  \centering
\begin{tikzpicture}
   \begin{axis}[
     articleplot,
      xlabel={$d$},
      ylabel={bits},
      width=8cm,
      height=6cm,
      xmin=0, xmax=1,
      ymin=0, ymax=1]
      \addplot[no marks, line width=0.8pt, white!30!purple] table[x index=0, y index=6] {data/markov_2.dat}; \label{plot:bm_site_r2}
      \addplot[no marks, line width=1.65pt, white!30!purple] table[x index=0, y index=7] {data/markov_2.dat}; \label{plot:bm_sdte_r2}
      \addplot[no marks, line width=0.8pt, black,  dash pattern=on 2pt off 1pt] table[x index=0, y index=10] {data/markov_2.dat}; \label{plot:bm_site2}
      \addplot[no marks, line width=1.65pt, black,  dash pattern=on 2pt off 1pt] table[x index=0, y index=11] {data/markov_2.dat}; \label{plot:bm_sdte2}
   \end{axis}
\end{tikzpicture}
  \caption{Decomposition of transfer entropy $T_{Z \to (X,Y)}$ for the second example process. The plot shows SITE (\ref{plot:bm_site2} using $\Imin$, \ref{plot:bm_site_r2} using $\Ired$) and SDTE (\ref{plot:bm_sdte2} using $\Imin$, \ref{plot:bm_sdte_r2} using $\Ired$).}
  \label{fig:binary_markov2}
\end{figure}

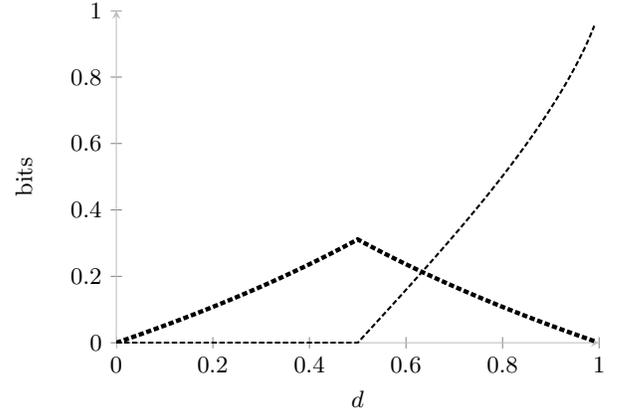
\begin{figure}
  \centering
\begin{tikzpicture}
   \begin{axis}[
     articleplot,
      xlabel={$d$},
      ylabel={bits},
      width=8cm,
      height=6cm,
      xmin=0, xmax=1,
      ymin=0, ymax=1]

      \addplot[no marks, line width=0.8pt, black,  dash pattern=on 2pt off 1pt] table[x index=0, y index=10] {data/markov_3.dat}; \label{plot:bm_site3}
      \addplot[no marks, line width=1.65pt, black,  dash pattern=on 2pt off 1pt] table[x index=0, y index=11] {data/markov_3.dat}; \label{plot:bm_sdte3}
   \end{axis}
\end{tikzpicture}
  \caption{Decomposition of transfer entropy $T_{Z \to Y}$ for the second example process. The plot shows SITE (\ref{plot:bm_site3} using $\Imin$), SDTE (\ref{plot:bm_sdte3} using $\Imin$).}
  \label{fig:binary_markov3}
\end{figure}

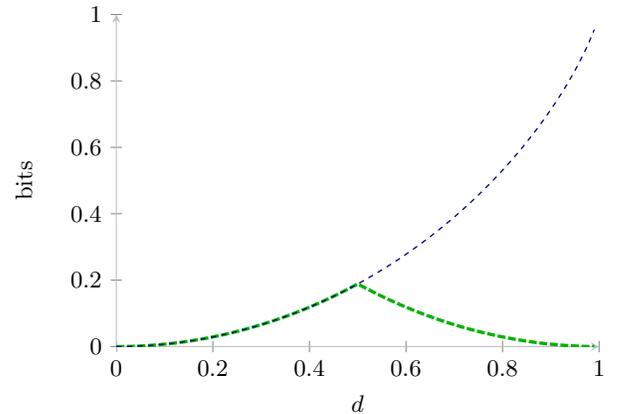
\begin{figure}
  \centering
\begin{tikzpicture}
   \begin{axis}[
     articleplot,
      xlabel={$d$},
      ylabel={bits},
      width=8cm,
      height=6cm,
      xmin=0, xmax=1,
      ymin=0, ymax=1]
      \addplot[no marks, line width=1.25pt, black!30!green, dash pattern=on 3pt off 1pt] table[x index=0, y index=4] {data/markov_2.dat}; \label{plot:bm_red}
      \addplot[no marks, line width=0.5pt, black!50!blue, dash pattern=on 2pt off 2pt] table[x index=0, y index=2] {data/markov_2.dat}; \label{plot:bm_mi}

   \end{axis}
\end{tikzpicture}
  \caption{The plot shows $I(Y_{t+1};Z_t)$ (\ref{plot:bm_mi}) and $\Ired(Y_{t+1};Y_t,Z_t)$ (\ref{plot:bm_red}) for the second example process.}
  \label{fig:binary_markov2b}
\end{figure}

The increase of mutual information $I(Y_{t+1};Z_t)$ is obvious from the definition of the process. For $d=0$ we have independence between both processes and for $d=1$ we have $Y_{t+1}=Z_t$. It is also clear that the redundant information with respect to $Y_{t+1}$  needs to be zero at the extremal points $d\in\{0,1\}$, because at these points the value of $Y_{t+1}$ depends either on $Y_t$ ($d=0$) or $Z_t$ ($d=1$) and therefore either $I(Y_{t+1};Z_t)=0$ or $I(Y_{t+1};Y_t)=0$ which both are upper bounds for the redundancy. 

On the other hand for $d=0.5$ the state of either process at time $t$ tells us something about the distribution of $Y_{t+1}$ and because the space of distributions of $Y_{t+1}$ is one-dimensional, this must be information about a change in the same direction, so there is positive redundancy. Observing one of the outcomes necessarily contributes to some extent to the prediction of the outcome of $Y_{t+1}$. We can now show this more rigourously, we have
\begin{eqnarray}
  p(y_{t+1}|y_t) &=& \frac{d}{2} \delta_{y_{t+1}  (1-y_t)} + \left( 1-\frac{d}{2} \right) \delta_{y_{t+1} y_t}, \\
  p(y_{t+1}|z_t) &=& \frac{1-d}{2} \delta_{y_{t+1} (1-z_t)} + \frac{1+d}{2} \delta_{y_{t+1} z_t}.
\end{eqnarray}
as the conditional distributions given the current state of either $Y_t$ or $Z_t$. To calculate $\Ired(Y_{t+1};Y_t,Z_t)$ we need to calculate the projected information $\Ipr{Z_{t}}{Y_t}{Y_{t+1}}$ and $\Ipr{Y_t}{Z_t}{Y_{t+1}}$ as the redundancy is the minimum of both terms. Because the space of distributions $\Delta(Y_{t+1})$ is one dimensional (it is simply the unit interval) we can make a simple illustrative argument to compute $\qpr{z_t=0}{Y_t},\qpr{z_t=1}{Y_t},\qpr{y_t=0}{Z_t}$ and $\qpr{y_t=1}{Z_t}$, which are the terms that are needed to calculate projected information.
From the illustration in FIG.~\ref{fig:ytspace} it can be seen that for $d\leq 0.5$,  $\qpr{z_t=0}{Y_t}(y_{t+1})=\qpr{y_t=0}{Z_t}(y_{t+1})=p(y_{t+1}|z_t=0)$ and $\qpr{z_t=1}{Y_t}(y_{t+1})=\qpr{y_t=1}{Z_t}(y_{t+1})=p(y_{t+1}|z_t=1)$. If we insert this into (\ref{eq:projective_info}) we get that $\Ipr{Z_{t}}{Y_t}{Y_{t+1}}=\Ipr{Y_t}{Z_t}{Y_{t+1}}=I(Y_{t+1};Z_t)$ for $d\leq 0.5$.

\begin{figure*}[t]
  \subfloat[$\Delta(Y_{t+1})$ for $d\leq 0.5$]{\label{fig:dleqhalf}
  \begin{tikzpicture}[distnode/.style={
        shape=diamond, style=fill, inner sep=1.5pt,
      }]

    \draw [square-square] (0,0) -- (10cm,0) node[pos=0,above]{\scriptsize $p(y_{t+1}=0)=1$}
     node[pos=1,above]{\scriptsize $p(y_{t+1}=1)=1$};
    \fill (5cm,0) node[distnode] {} node[below] {\scriptsize $p(y_{t+1})$};
    \draw[red!10] (1.25cm,0) node[distnode] {} node[below, black!50!red] {\scriptsize $p(y_{t+1}|y_t=0)$};
    \draw[red!10] (8.75cm,0) node[distnode] {} node[below, black!50!red] {\scriptsize $p(y_{t+1}|y_t=1)$};
    \fill[blue!50] (3.75cm,0) node[distnode] {} node[above, black!50!blue] {\scriptsize $p(y_{t+1}|z_t=0)$};
    \fill[blue!50] (6.25cm,0) node[distnode] {} node[above, black!50!blue] {\scriptsize $p(y_{t+1}|z_t=1)$};
    \fill[blue!50!red] (3.75cm,0) node[distnode, draw, fill=none] {}  node[below=0.5cm] {\scriptsize $\qpr{z_t=0}{Y_t}=\qpr{y_t=0}{Z_t}$};
    \fill[blue!50!red] (6.25cm,0) node[distnode, draw, fill=none] {} node[below=1.0cm] {\scriptsize $\qpr{z_t=1}{Y_t}=\qpr{y_t=1}{Z_t}$};
    \draw[blue!50!red] (3.75cm,0) -- (3.75cm,-0.5cm);
    \draw[blue!50!red] (6.25cm,0) -- (6.25cm,-1.0cm);
  \end{tikzpicture}}\\
    \subfloat[$\Delta(Y_{t+1})$ for $d\geq 0.5$]{\label{fig:dgeqhalf}
  \begin{tikzpicture}[distnode/.style={
        shape=diamond, style=fill, inner sep=1.5pt
      }]
    \draw [square-square] (0,0) -- (10cm,0) node[pos=0,above]{\scriptsize $p(y_{t+1}=0)=1$}
     node[pos=1,above]{\scriptsize $p(y_{t+1}=1)=1$};
    \fill (5cm,0) node[distnode] {} node[below] {\scriptsize $p(y_{t+1})$};
    \fill[red!10] (3.75cm,0) node[distnode] {} node[above,black!50!red] {\scriptsize $p(y_{t+1}|y_t=0)$};
    \fill[red!10] (6.25cm,0) node[distnode] {} node[above,black!50!red] {\scriptsize $p(y_{t+1}|y_t=1)$};
    \fill[blue!50] (1.25cm,0) node[distnode] {} node[below,black!50!blue] {\scriptsize $p(y_{t+1}|z_t=0)$};
    \fill[blue!50] (8.75cm,0) node[distnode] {} node[below,black!50!blue] {\scriptsize $p(y_{t+1}|z_t=1)$};

    \fill[blue!50!red] (3.75cm,0) node[distnode, draw, fill=none] {} node[below=0.5cm] {\scriptsize $\qpr{z_t=0}{Y_t}=\qpr{y_t=0}{Z_t}$};
    \fill[blue!50!red] (6.25cm,0) node[distnode, draw, fill=none] {}  node[below=1.0cm] {\scriptsize $\qpr{z_t=1}{Y_t}=\qpr{y_t=1}{Z_t}$};
    \draw[blue!50!red] (3.75cm,0) -- (3.75cm,-0.5cm);
    \draw[blue!50!red] (6.25cm,0) -- (6.25cm,-1.0cm);
  \end{tikzpicture}}
  \caption{Illustration of the conditional distributions of $Y_{t+1}$ for the second example process in the two cases $d\leq 0.5$ and $d\geq 0.5$. The line represents the one dimensional simplex, i.e. the space of probability distributions over $Y_{t+1}$ denoted by $\Delta(Y_{t+1})$ where $Y_{t+1}$ is a binary valued random variable. The black diamond represents the marginal distribution of $p(y_{t+1})$ and the shaded diamonds the conditionals given specific values of $Y_t$ and $Z_t$. It can now be seen that the projections are always equal to the conditional distributions closer to the marginal of $Y_{t+1}$. In particular, the projections are the same, no matter in which direction the projection is done (from $Y_t$ to $Z_t$ or vice versa). }
  \label{fig:ytspace}
\end{figure*}
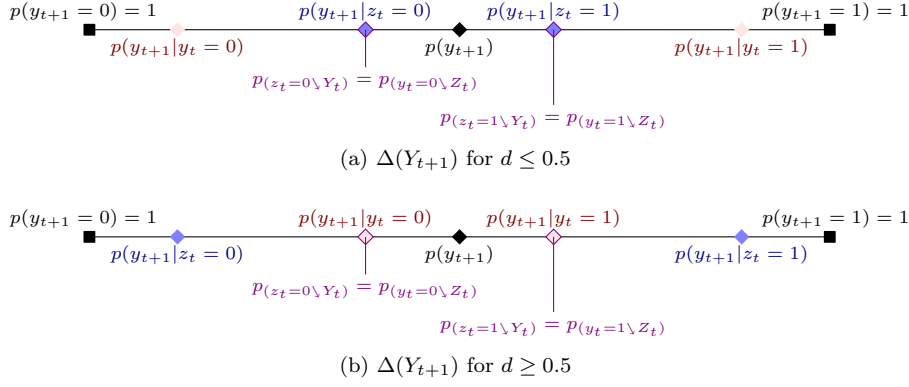

Conversely for $d\geq 0.5$ we get $\Ipr{Z_{t}}{Y_t}{Y_{t+1}}=\Ipr{Y_t}{Z_t}{Y_{t+1}}=I(Y_{t+1};Y_t)$ for $d\leq 0.5$. As $I(Y_{t+1};Z_t)$ and $I(Y_{t+1};Y_t)$ are perfectly symmetric, this then explains the form of the redundant information as in (\ref{plot:bm_red} in FIG.~\ref{fig:binary_markov2b}). Thus, even though $Z_t$ and $Y_t$ are completely independent, the mechanism, which is a random read-out (with distribution $d$,$(1-d)$), creates redundancy with respect to $Y_{t+1}$. Furthermore, this explains why we have no state-independent transfer entropy for $d\leq 0.5$.

\subsubsection{Open Loop Controllability}

\citet{Ashby1956} proposed  and \citet{touchette2000} confirmed  that there is a natural link between control theory and information theory. As shown by \citet{Touchette2004}, for a process, with initial state $X$ and final state $X'$, and a controller $C$ which are linked by the probability distribution $p(x'|x,c)$, the conditional mutual information $I(X';C|X)$ (which is the transfer entropy from the controller to the system) is a measure of controllability. Williams and Beer show in \cite{Williams2011} that the decomposition of transfer entropy using $\Imin$ as a redundancy measure has a close relation to the notion of open-loop controllability. We will now show, that this is still the case if $\Ired$ is used to decompose transfer entropy.

Perfect controllability, as defined in \cite{Touchette2004}, means that for all initial states $x\in\mathcal{X}$ and final states $x' \in \mathcal{X}$ there exists a control state $c\in \mathcal{C}$ such that $p(x'|x,c)=1$. The following equivalence is then shown in \cite{Williams2011}
\begin{lemma}
  \label{lem:pceq}
  A system is perfectly controllable iff for any $x'$ there exists a distribution $p(c|x)$ such that $p(x')=1$ for any distribution $p(x)$.
\end{lemma}


It follows also that if a system is perfectly controllable, there exists an $x'$ such that $p(x'|x)=1$ for each $x\in \mathcal{X}$, see \cite{Williams2011} for a proof. Now, a system has perfect open-loop controllability iff it has perfect controllability and $I(X;C)=0$. Moreover, in \cite{Williams2011} it is shown that the following theorem holds:
\begin{theorem}[Williams and Beer] \label{thm:open_loop}
  A system is perfectly open-loop controllable iff it is perfectly controllable with vanishing state-dependent transfer entropy (using $\Imin$) from $C$ to $X'$.
\end{theorem}
We will now also show that this theorem still holds in the case where the decomposition using our measure of redundant information $\Ired$ is used. To prove the theorem we will use the following lemma. It is shown in \cite{Williams2011} that the condition of the lemma is fulfilled for any perfect open-loop controller and thus proves the direct part of the theorem (perfect open-loop controllability implies perfect controllability with zero SDTE using $\Ired$ as a redundancy measure):

\begin{lemma}
If
\begin{equation*}
  p(x'|x,c) = p(x'|c) \quad  x' \in \mathcal{X}, \forall x \in \mathcal{X}, c \in \mathcal{C}
\end{equation*}
then the STDE from $C$ to $X'$ is zero.
\end{lemma}

\begin{proof}
  From (\ref{eq:redundancy_def}) and (\ref{eq:pia_decomp}) it follows that
  \begin{eqnarray}
    \pia{}{(X';\{C,X\}) } &\leq& I(X';X,C)-I(X';X) \nonumber \\*
    & & -I(X';C)+\Ipr{X}{C}{X'}.
  \end{eqnarray}
  The synergy is non-negative and now the right hand side can be reformulated as in (\ref{eqn:before_convex}). But with $p(x'|x,c) = p(x'|c)$ $\forall x,x' \in \mathcal{X}, c \in \mathcal{C}$ the positive Kullback-Leibler divergences in (\ref{eqn:before_convex}) all vanish. Therefore $\pia{}{(X';\{C,X\}) }=0$.
\end{proof}
For the converse direction, perfect controllability and vanishing STDE (from $C$ to $X'$) imply perfect open-loop controllability, we first need to prove the following lemma:

\begin{lemma}
  \label{lem:perfectredun}
  If a system is perfectly controllable with a distribution $p(c|x)$ then $\Ired(X';X,C)=0$.
\end{lemma}

\begin{proof}
  From Lemma~\ref{lem:pceq} it follows that $p(x')=1$ for some $x'\in \mathcal{X}$ as well as $p(x'|x)=1$ for all $x\in \mathcal{X}$ and therefore $\Ccl(\langle X \rangle_Z)$ in $\Delta(X')$ is just $\{p(x')\}$ which implies $\Ipr{C}{X}{X'}=0$. Thus it follows that $\Ired(X';X,C)=0$.
\end{proof}

Thus, for the converse direction, starting with perfect controllability and vanishing STDE, we have the following equality
\begin{align}
  0 &= \pia{}{(X';\{C,X\}) } \\
  &= I(X';X,C)-I(X';X) \nonumber \\*
  &  -I(X';C)+\Ired(X';X,C)\\
  &= I(X';X,C)-I(X';X)-I(X';C)\\
  &= \sum_{x,c,x'} p(x',x,c) \log \frac{p(x'|x,c)p(x')}{p(x'|c)p(x'|x)},\\
  \intertext{as we also have $p(x'|x)=p(x')$ because of perfect controllability,}
  &= \sum_{x,c,x'} p(x',x,c) \log \frac{p(x'|x,c)}{p(x'|c)}.
\end{align}
We also know that for every $x \in \mathcal{X}$ there exists $x'\in \mathcal{X}$ and $ c \in \mathcal{C}$ such that $p(x'|x,c)=1$. Thus for any $x'\in \mathcal{X}$ there exists a $ c \in \mathcal{C}$ such that $p(x'|c)=1$. It is shown in \cite{Williams2011} that this is equivalent to open-loop controllability. 

Hence, we have shown that Theorem~\ref{thm:open_loop} also holds if we apply $\Ired$ as the underlying redundancy measure and the relation between open-loop controllability and decomposition of transfer entropy is transferable to our new measure.

\section{Discussion}

The motivation for this paper was to overcome the shortcomings of current measures of redundancy and synergy. We introduced a new measure for bivariate redundant information. Redundant information between two random variables is information that is shared between two variables. In contrast to mutual information, redundant information denotes information with respect to the outcome of a third variable. Our measure is conceptually motivated by measuring similarities in the direction of change in the outcome distribution, depending on which input is observed. We proved that the construction adheres to properties of redundancy as stated in the literature, and can be used for a non-negative decomposition of mutual information. The measure is closely related to the concept of \textit{minimal information} as introduced in \cite{Williams2010}.

We demonstrated in several examples that $\Ired$ follows several intuitions about redundancy. Furthermore, it is possible to decompose \textit{transfer entropy} as considered in \cite{Williams2011}; in particular we showed that using \textit{minimal information} instead of \textit{redundant information} to decompose \textit{transfer entropy} can lead to the detection of fake state-dependent transfer entropy.  We were able to prove that the results about open-loop controllability from \cite{Williams2011} are also applicable to the decomposition using $\Ired$. Thus our measure is able to serve as a replacement for the bivariate version of minimal information.

A particular insight of our definition is the emphasis of mechanisms in the concept of redundant information, which has been rather neglected in the literature so far. Firstly, we linked bivariate redundant information in the case of a copying mechanism to the mutual information between the input variables. We identify redundant information that already appears in the inputs with \textit{source redundancy}, contrary to redundant information that is only due to the mechanism, as demonstrated in the \textsc{And}-gate or the 50:50-readout. We identify this kind of redundancy with \textit{mechanistic redundancy}. This is in contrast to the redundancy measure proposed in \cite{Griffith2012} which does not capture \textit{mechanistic redundancy}. The separation of both kinds of redundancy is not explicit at this point, and currently we do not yet propose a clear and obvious separation of mechanistic and source contributions of redundant information. 

Future work will show whether it is possible to separate the two concepts of mechanistic and source redundancy when they appear simultaneously. Another limitation we currently have is the restriction to a bivariate measure. In general, however, there are applications where it is interesting to be able to compute redundant information between more than two variables \cite{Williams2010,Flecker2011}. However, the geometric structure for this problem gets significantly more complex, and it is, for example, not entirely clear by what the \textit{identity property} should be replaced in the multivariate case. There are several ways to generalize mutual information to a multivariate measure, none of which seems to be fitting in this case. The construction of a multivariate measure of redundant information, as well as a generalization to continuous random variables is part of ongoing research.

\begin{acknowledgments}
DP thanks Virgil Griffith for helpful discussions. This research was partially supported (CS and DP) by the European Commission as part of the CORBYS (Cognitive Control Framework for Robotic Systems) project under contract FP7 ICT-270219. The views expressed in this paper are those of the authors, and not necessarily those of the consortium.
\end{acknowledgments}

\bibliographystyle{apalike}
\bibliography{redundancy}

\end{document}